\documentclass[aps, pra, a4paper, longbibliography, twocolumn]{revtex4-1}

\usepackage[utf8]{inputenc}
\usepackage[british]{babel}  
\usepackage[osf,sc]{mathpazo}
\usepackage{graphicx} 
\usepackage{xcolor}
\definecolor{darkred}{rgb}{.8 .1 .1}
\definecolor{darkgreen}{rgb}{.1 .8 .1}
\definecolor{darkyellow}{rgb}{.6 .6 .0}

\usepackage[colorlinks,citecolor=blue,linkcolor=blue,urlcolor=blue,linktocpage=true]{hyperref}

\usepackage{amsmath}
\usepackage{amssymb,amsthm,bm,mathtools,amsfonts,mathrsfs,bbm,dsfont}

\newcommand{\mm}[1]{\begin{align} #1 \end{align}}

\newcommand{\C}{\mathbb{C}}

\newcommand{\ket}[1]{\ensuremath\,|#1\rangle}

\newcommand{\ketbra}[2]{|#1\rangle \langle #2 |}

\newcommand{\tr}[2][]{\text{Tr}_{#1}\left(#2\right)}
\renewcommand{\rho}{\varrho}

\newcommand{\ot}{\otimes}

\newcommand{\id}{\mathds 1}
\newcommand{\idd}{\text{id}}

\newcommand{\sgn}{\text{sgn}}

\newcommand{\sS}{\mathbb{S}}
\newcommand{\NN}{\mathbb{N}}

\renewcommand{\eqref}[1]{Eq.~(\ref{#1})}
\renewcommand{\refeq}[1]{Eq.~(\ref{#1})}

\newcommand{\fpi}[2][A]{#2_#1}

\theoremstyle{definition}
\newtheorem{theorem}{Theorem}
\newtheorem{lemma}[theorem]{Lemma}
\newtheorem{corollary}[theorem]{Corollary}
\theoremstyle{remark}

\begin{document}

\title{
Multiparticle singlet states cannot be maximally entangled 
for the bipartitions
}

\author{Fabian Bernards}
\affiliation{Naturwissenschaftlich-Technische Fakultät, 
Universität Siegen, Walter-Flex-Straße 3, 57068 Siegen, Germany}

\author{Otfried Gühne}
\affiliation{Naturwissenschaftlich-Technische Fakultät, 
Universität Siegen, Walter-Flex-Straße 3, 57068 Siegen, Germany}

\date{\today}

\begin{abstract}
One way to explore multiparticle entanglement is to ask for 
maximal entanglement with respect to different bipartitions, 
leading to the notion of absolutely maximally entangled 
states or perfect tensors. A different path uses unitary 
invariance and symmetries, resulting in the concept of 
multiparticle singlet states. We show that these two 
concepts are incompatible in the sense that 
the space of pure multiparticle singlet states does not contain 
any state for which all partitions of two particles versus 
the rest are maximally entangled. This puts restrictions on the 
construction of quantum codes and contributes to discussions
in the context of the AdS/CFT correspondence and quantum gravity.
\end{abstract}
\maketitle

\section{Introduction}
The notion of multiparticle entanglement is relevant for different 
fields in physics, such as condensed matter physics or quantum 
information processing. A fundamental problem, however, lies in 
the exponential scaling of the dimension of the underlying Hilbert 
space, rendering an exhaustive classification difficult. So, in 
order to gain insight into multiparticle entanglement phenomena 
as well as to identify potentially interesting quantum states, 
different concepts based on symmetries \cite{werner, eggeling, symm1, symm2, symm3}, 
graphical representations \cite{graph1, graph2, graph3}, matrix product approximations \cite{mps1, mps2} or entanglement quantification \cite{maxe1, maxe2, maxe3, maxe4, maxe5, maxe6} 
can be used.

Indeed, it is a natural question to ask for states with maximal 
entanglement. A pure two-particle quantum state is maximally 
entangled if the reduced state for one particle is maximally 
mixed \cite{vidal}. One can extend this definition to the multiparticle 
case by considering bipartitions of the particles into two groups 
and asking whether the global state is maximally entangled
for the bipartitions. A multiparticle state is then absolutely 
maximally entangled (AME) if is it maximally entangled for all 
bipartitions. A weaker requirement is that all bipartitions of 
$k$ particles versus the rest shall be maximally entangled, these 
states are called $k$-uniform.

As AME states and $k$-uniform states are central for quantum error 
correction, they are under intensive research \cite{scott, grassl, felix}. 
Indeed, some recently solved questions concerning the sheer existence of AME 
states have been highlighted as central problems in quantum 
information theory \cite{karol1, karol2}. In addition to their use in quantum 
information processing, AME states, under the name of perfect tensors, 
have been used to construct toy models for the AdS/CFT correspondence 
\cite{pastawski, bhattacha} and to study the entanglement entropy in conformal 
field theories \cite{ryu}.

A different concept to explore the phenomena of multiparticle 
entanglement is the use of symmetries. Here, the study of unitary 
invariance has some tradition \cite{werner, eggeling, adan-singlets, huberdimension}. Consider the two-qubit 
singlet state 
\begin{equation}
 \ket{\psi^-} = \frac{1}{\sqrt{2}}(\ket{01}-\ket{10}).
 \label{eq-2qbsinglet}
\end{equation}
This is not only a maximally entangled state (as the reduced 
density matrices are maximally mixed), it also has the property 
that it is form-invariant under simultaneous application of local 
unitaries, that is,  $U \otimes U \ket{\psi^-} = e^{i \phi} \ket{\psi^-}$. 
Historically, this property was the key to understand the difference 
between quantum entanglement and the violation of Bell inequalities 
\cite{werner}.

For more particles, one finds more states with this kind of unitary 
symmetry, and such states are called general singlet states. 
For three three-level systems, there is the totally 
antisymmetric state
$
\ket{\psi_3} = 
(\ket{012}+\ket{201}+\ket{120}
- \ket{210}- \ket{102}-\ket{021}
)/\sqrt{6},
$
and for four qubits, the unitarily invariant states form a 
two-dimensional subspace \cite{rumer1, rumer2, pauling, beach}, spanned by 
a two-copy extension of the singlet state (\ref{eq-2qbsinglet}) 
and the four-qubit singlet state \cite{kempe, bourenanne}
\begin{align}
 \ket{\psi_4} &= \frac{1}{\sqrt{3}}
 \big[
 \ket{0011}
 +
 \ket{1100}
 \nonumber \\
 & -\frac{1}{2} (\ket{01}+ \ket{10}) \otimes (\ket{01}+ \ket{10})
 \big].
 \label{eq-4qbsinglet}
\end{align}
This two-dimensional subspace can be used to encode a quantum bit such
that it is immune against collective decoherence \cite{lidar, kempe, bourenanne}. 
In addition, 
multiparticle singlet states have turned out to be useful for various 
quantum information tasks, such as secret sharing and liar detection, 
see Ref.~\cite{adan-singlets} for an overview. Finally, states with unitary 
symmetry have attracted interest from the perspective of quantum gravity 
\cite{gravity1, gravity2, gravity3, sahlmann}.

It is a natural question to ask whether these two research lines are
connected, in the sense that one can find $k$-uniform states or even
AME states which belong the the invariant subspace of singlet states.
This question was first studied in Ref.~\cite{grassl-njp} where it 
has been shown that for four particles and arbitrary dimensions no 
AME singlet state can be found. Using methods from quantum gravity 
it was shown recently in Ref.~\cite{sahlmann} that for six qubits 
the singlet states cannot be AME. Combined with known results on the 
non-existence of AME states, this shows that multi-qubit singlet states 
cannot be AME (see also \cite{fabian-phd}).

In this paper we prove that the space of pure multiparticle 
singlet states of any particle number and any dimension does not
contain states which are two-uniform. Consequently, there are 
no AME states in this subspace for four or more particles of any dimension. 
From the viewpoint of quantum information processing this may 
be interpreted as showing that no code words of the usual quantum 
codes are within the singlet subspace, or that quantum code words 
cannot be immune against collective decoherence. Our proof technique 
seems conceptually simpler than previous approaches 
\cite{grassl-njp, sahlmann}.  
In order to be understandable for the different communities, we 
formulate in the following our approach self-contained and 
in mathematical terms.

\section{Properties of pure singlet states}

First, a {\it pure multiparticle singlet state} (also called pure Werner state)
is a quantum state 
$\ket{\psi} \in (\C^d)^{\ot n}$ such that for any 
unitary $U \in U(d)$ it holds that
\begin{equation}
U \ot \ldots \ot U \ket{\psi} = \xi \ket{\psi}, 
\label{eq-purewerner}
\end{equation}
where $\xi \in \C$ is a complex phase that depends on the 
unitary $U$ and the state. The function 
$\xi: U(d) \to \C$ that satisfies
\begin{equation}
U \ot \ldots \ot U \ket{\psi} = \xi(U) \ket{\psi}, 
\label{eq-phasefunc}
\end{equation}
is called the {\it phase function} of $\ket{\psi}$.
Clearly, the phase function $\xi$ is a group homomorphism, that is 
$ \xi(U_1 U_2) = \xi(U_1) \xi(U_2)$ and $\xi(\id) = 1$.
The phase function determines the action of unitaries on
a singlet state and will be the key to study properties
of singlet states for our purpose. In the following, we 
will consider the behaviour of this function for two 
special types of unitaries, local permutations of the 
basis vectors, and unitaries diagonal in the computational
basis.

\subsection{The phase function for local permutations}

In order to fix our notation, let us 
first note basic facts about the phase function 
in the case of permutations.

First, we take $\{\ket k\}$ for ${k \in \{0,...,d-1\}}$ as
the computational basis of $\C^d$ and consider unitary 
matrices that induce a permutation of these basis vectors. 
More formally, we consider the representation $(V, \C^d)$ of 
the symmetric (permutation) group $\sS_d$ which is defined as
\begin{equation}
V(\pi) \ket{k} = \ket{\pi(k)}
\end{equation}
and call it the {\it canonical representation} of $\sS_d$.

Second, let $(V, \C^d)$ be the canonical representation of 
$\sS_d$ and let $\xi$ be the phase function of the singlet 
state $\ket{\psi}$. Then we call
\begin{equation}
f (\pi) = \xi\big[V(\pi)\big]
\end{equation}
the {\it permutation phase function} of $\ket{\psi}$. The following 
Lemma puts a strong restriction to the permutation phase function 
of any singlet state.

\begin{lemma}
There are only two functions that are compatible with the definition 
of the permutation-phase function $f$, namely either 
$f(\pi) = 1, \, \forall \pi \in \sS_d$ or $f(\pi) = \sgn(\pi)$, 
the signum of the permutation.
\label{lemma1}
\end{lemma}

\begin{proof}
Since $\xi$ is a group homomorphism, $f$ is a group homomorphism, 
too. So we only need to consider the images of the generators of 
$\sS_d$ under $f$ in order to characterize $f$.

The symmetric group $\sS_d$ is generated by the transpositions
$\theta_k = (k, k+1)$ of two neighboring elements with $k \in \{0,...,d-2\}$, 
so $f$ is determined by $f(\theta_k)$. Further, the transpositions fulfill 
the following relations
\begin{align}
\theta_k^2 &= \idd \label{eq-transpid}\\
\theta_k \theta_m &= \theta_m \theta_k, \text{ for } |k-m|\ge 2 \\
\theta_k \theta_m \theta_k &= \theta_m \theta_k \theta_m, \text{ for } |k-m| =
1. \label{eq-transprel3}
\end{align}
Since $f$ is a homomorphism, \refeq{eq-transpid} gives 
$f(\theta_k) = \pm 1$ and \refeq{eq-transprel3} gives
\begin{equation}
f(\theta_k) = f(\theta_m) \text{ for } |k-m| = 1.
\end{equation}
Hence, we only have to distinguish two cases: Either $f(\theta_k) = 1 \quad \forall k$
which implies that $f\equiv1$ or (2) $f(\theta_k) = -1 \quad\forall k$, which means that 
$f(\pi) = \sgn(\pi).$
\end{proof}

Lemma \ref{lemma1} allows  to derive a simple relation between some of 
the coefficients of a singlet state. We can state the result more elegantly 
by making use of a multi-index notation. We define a {\it multi-index} $i$ 
as a sequence of indices $i_{\alpha} \in \{0,\ldots, d-1\}$ with 
$\alpha \in \{1, \ldots, n\}$. In this notation, the action of the group 
$\sS_d$ on a multi-index $i = (i_1, \ldots, i_n)$ can be written as
\begin{equation}
\pi(i) = \big(\pi(i_1), \ldots, \pi(i_n)\big).
\end{equation}
Then we have the following Lemma.

\begin{lemma}
\label{lemma-permutationd}
Let $\ket{\psi} = \sum_i t_i \ket i \in (\C^d)^{\otimes n}$ be a pure 
singlet state and $\pi \in \sS_d$ a permutation. Then it holds that
\mm{
t_{\pi(i)} = f(\pi) t_i
}
where $f$ is the permutation-phase function of $\ket{\psi}$.
\end{lemma}

\begin{proof}
Taking the identity 
$
V(\sigma)^{\otimes n} \ket{\psi} = f(\sigma) \ket{\psi} 
$
and comparing the coefficients in front of the basis vector $\ket{i}$
results in
\begin{equation}
t_{\sigma^{-1}(i)}  = f(\sigma) t_i. 
\end{equation}
Choosing $\sigma = \pi^{-1}$ and using the fact that $f(\sigma)=f(\sigma^{-1})$
(which follows from $f(\sigma) = \pm 1$) proves the claim.
\end{proof}


\subsection{The phase function for diagonal unitaries}
We now characterize the phase function further by considering 
the restriction of $\xi$ to the diagonal unitaries. The key result
is that for a multiparticle singlet state expanded
in the computational basis many coefficients  have to vanish.
We start with a technical lemma.

\begin{lemma}
Let $\ket{\psi} = \sum_i t_i \ket{i}$ be a pure $n$-particle singlet state. Then, for any 
$k \in \{0, \ldots, d-1\}$, there exists a number $N_k \in \NN$ such that for
any non-zero coefficient $t_l \neq 0$, the multi-index $l$ contains the value
$k$ exactly $N_k$ times.

In other words, the non-vanishing terms of a singlet state in the computational
basis are tensor products $\ket{i_1, i_2, i_3, \dots, i_n}$ of $N_0$ times the 
single-particle state $\ket{0}$, $N_1$ times the single-particle state $\ket{1}$, 
etc. This naturally implies $\sum_{k=0}^{d-1} N_k =n$.
\label{lemma-samenumber}
\end{lemma}

\begin{proof}
We consider \refeq{eq-purewerner} for diagonal unitary matrices.
Diagonal unitary matrices are generated by the matrices
\begin{equation}
U_k = \id + (e^{i\phi_k} - 1) \ketbra{k}{k},
\end{equation}
where all diagonal entries, except the $k$-th one, are $1$,
while the $k$-th diagonal entry is $e^{i\phi_k}$. Such a 
diagonal matrix acts on a single-particle computational basis 
vector $\ket a$ as
\begin{equation}
U_k \ket a = \begin{cases} e^{i\phi_k} \ket a, \text{ if } a = k, \\ \ket a, \,
\text{else.} \end{cases}
\end{equation}
We now determine the phase function $\xi(U_k)$ from \refeq{eq-purewerner}. This yields
\begin{align}
U_k^{\otimes n}\ket{\psi} 
= 
\sum_i t_i U_k \ket{i_1} \ot \ldots \ot U_k \ket{i_n} 
=  \sum_i e^{i\phi_k K_i} t_i  \ket i,
\end{align}
where $K_i$ is the number of times that $k$ appears in the multi-index $i$. 
This equation holds for arbitrary $\phi_k$; moreover, $\xi(U_k)$ affects all 
coefficients $t_j \neq 0$ in the same way. Hence, all indices $l$, for which 
$t_l \neq 0$ must contain $k$ the same number of times.
\end{proof}

In the following, we combine this result with the previous results on
local permutations of the basis vectors. Intuitively, one may  expect
that the numbers $N_k$ in Lemma \ref{lemma-samenumber} can not depend on $k$, 
as one can always change the index $k$ by local permutations, without
affecting the singlet state too much. This will indeed turn out to be 
the case.

Before formulating this in a precise manner, note that we can 
describe the action of a permutation $\omega \in \sS_n$ of the
particles on a multi-index as
\begin{equation}
\omega(i) = (i_{\omega(1)}, ..., i_{\omega(n)}).
\end{equation}
In this language, if $\ket{\psi} = \sum_i t_i \ket i$ is a pure singlet state
where $j,l$ are two multi-indices such that  $t_j \neq 0 $ and $t_l\neq 0$. 
Then there exists a permutation $\omega \in \sS_n$ of the particles such that 
$j = \omega(l)$. This follows directly from Lemma \ref{lemma-samenumber}.
More generally, we have:

\begin{lemma}
\label{lemma-samenumber2}
Let $\ket{\psi} = \sum_i t_i \ket i$ be a pure singlet state. Then 
there 
exists a number $K \in \NN$, such that for any non-zero coefficient $t_l \neq 0$, 
the multi-index $l$ contains each value $k \in \{0, \ldots, d-1\}$ exactly $K$ 
times.
\end{lemma}
\begin{proof}
Consider two values $k, m \in \{0, \ldots, d-1\}$ and let $N_k $ and $N_m$ be 
defined as in Lemma \ref{lemma-samenumber}. So, $N_k$ denotes how many times 
$k$ is contained in the multi-index $l$ of any coefficient $t_l \neq 0$. Now 
consider the local permutation $\pi= (k, m) \in \sS_d$, which just replaces
$k$ with $m$ and vice versa. According to Lemma \ref{lemma-permutationd}, 
we have
\begin{equation}
t_{\pi(l)} = f[(k,m)] t_{l}.
\end{equation}
Therefore, $t_l \neq 0$ implies $t_{\pi(l)} \neq 0$ due to Lemma 
\ref{lemma1}. However, $\pi(l)$ contains $k$ exactly $N_m$ times 
and $m$ exactly $N_k$ times. Thus, we must have $N_m = N_k$.
\end{proof}

It follows that for any pure singlet state $\ket{\psi} \in (\C^d)^{\otimes n}$, 
$n$ is always an integer multiple of $d$, that is $n = Kd$ for some $K \in \NN$.

In view of the previous results, the structure of the four-qubit singlet
state in Eq.~(\ref{eq-4qbsinglet}) becomes clearer now. For this state, 
we have $K=2,$ so it must be a superposition of tensor products with two 
$\ket{0}$ and two $\ket{1}$ factors. The phase function of two possible
permutations is just the identity. Note that Lemma \ref{lemma-samenumber2}
also implies that pure singlet states cannot exist for an odd number of qubits.

\section{Main results}
With the previous insights, we have characterized pure singlet states 
for our needs and we can now consider the second property, absolutely
maximally entangled states and $k$-uniform states. 

We start with introducing some notation on truncated 
multi-indices. Let $A \subset \{1, \ldots, n\}$ be a set 
of particles and $B = A^C =\{1, \ldots, n\} \setminus A$ be 
its complement. For a multi-index 
$i = (i_{\alpha})_{\alpha \in \{1, \ldots, n\}}$
we define the truncated
multi-indices
\begin{align}
i_A = (i_{\alpha})_{\alpha \in A}
\quad{\rm and} \quad
i_{B} = (i_{\alpha})_{\alpha \in B}.
\end{align}
Using this notation we express the marginal state of the subsystems 
contained in $A$ as 
\begin{align}
\varrho_A & =  \tr[B]{\ketbra{\psi}{\psi}}
=\sum_{i_A, j_A} \varrho_A (i_A; j_A)\ketbra{i_A}{j_A},
\end{align}
with the coefficients of the marginal density matrix
\begin{equation}
\varrho_A (i_A; j_A) = \sum_{i_B} t_{i_A,i_B} t^*_{j_A,i_B}.
\label{eq-element}
\end{equation}

Now we can define AME states and $k$-uniform states
in a very explicit manner that is useful for our purpose.
We call a state $\ket{\psi} \in (\C^d)^{\otimes n}$  a {\it $k$-uniform 
state}, if for any set of parties $A$ with $|A| \leq k$ 
the corresponding marginal is a maximally mixed state, that 
is
\begin{equation}
\varrho_A(i_A; j_A)  = \begin{cases} \frac 1{d^{|A|}} \text{, if } \fpi i =
\fpi j \\ 0 \text{, else.} \end{cases}
\end{equation}
A quantum state is {\it absolutely maximally entangled (AME)} if it is 
$k$-uniform 
with $k=\lfloor n/2 \rfloor$.

Note that for multiparticle singlet states the single-particle
marginals are always maximally mixed, so they are automatically 
$1$-uniform. This follows from the fact that the single-particle 
reduced state $\varrho_{\{1\}}$ is also unitarily invariant, that is
$U\varrho_{\{1\}}U^\dagger = \varrho_{\{1\}}$, and for a single particle
this implies that $\varrho_{\{1\}} = \openone /d$. But can singlet states 
be two-uniform?

We are ready to prove that this is never the case. To introduce the 
argument in a simplified setting, consider a putative six-qubit singlet 
state, which should be two-uniform. A six-qubit singlet state has terms 
consisting of three ``$0$'' and three ``$1$'' in the computational basis 
($K=3$ in Lemma 4). For a two uniform state, we should have $\varrho_{\{1,2\}} = \openone /4$, 
fixing the diagonal elements in any basis. So, for the diagonal element
$\varrho_{\{1,2\}}(0,0;0,0)$ in the computational basis it should hold that
\begin{align}
\varrho_{\{1,2\}}(0,0;0,0) &= |t_{(0,0,0,1,1,1)}|^2
+
|t_{(0,0,1,0,1,1)}|^2
\nonumber \\
&
+
|t_{(0,0,1,1,0,1)}|^2
+
|t_{(0,0,1,1,1,0)}|^2
\stackrel{!}{=} \frac{1}{4}.
\label{eq-6qb3}
\end{align}
If we consider also $\varrho_{\{1,2\}}(1,1;1,1)$ and sum over 
all qubit pairs $\alpha, \beta$ we obtain after a short 
calculation the condition
\begin{align}
\sum_{\alpha, \beta} \big[\varrho_{\{\alpha,\beta\}}(0,0;0,0)
+
\varrho_{\{\alpha,\beta\}}(1,1;1,1)\big]
\nonumber 
\\
=
2 {{3}\choose{2}}
\sum_i |t_i|^2 \stackrel{!}{=}\frac{1}{4}\times 2 \times {6 \choose 2} = \frac{15}{2},
\end{align}
where the prefactor $2 {{3}\choose{2}}=6$ can be understood
as coming from the fact that any $|t_i|^2$ contributes to 
$2 {{3}\choose{2}}$ terms of the form 
$\varrho_{\{\alpha,\beta\}}(\ell,\ell;\ell,\ell)$. But then, 
this condition is in direct contradiction to the normalization 
condition $\sum_i |t_i|^2=1$, so a six-qubit singlet state cannot
be two-uniform. 

This reasoning can directly be generalized to arbitrary dimensions
and numbers of particles. Indeed, in the general case with $K \geq 2$, 
we have
\begin{align}
\sum_{\alpha, \beta} 
\sum_{\ell=0}^{d-1}
\varrho_{\{\alpha,\beta\}}(\ell,\ell;\ell,\ell)
=
d {K \choose 2}
\sum_i |t_i|^2 \stackrel{!}{=}\frac{1}{d^2}\times d \times {n \choose 2},
\end{align}
where, as before, $K=n/d.$ Again, it can easily be seen that
this is in contradiction to the normalization 
$\sum_i |t_i|^2=1$. For the case $K=1$, however, one has 
$\varrho_{\{\alpha,\beta\}}(0,0;0,0)=0$ for any $\alpha, \beta$, 
consequently the marginals cannot be maximally mixed.
In summary, we have proved
the following: 

\begin{theorem}
Let $\ket{\psi} \in (\C^d)^{\otimes n}$ be a multiparticle singlet 
state. Then, not all two-particle reduced density matrices can
be maximally mixed. In other words, singlet states cannot be 
two-uniform.
\label{theorem-2uniform}
\end{theorem}

For $n\geq 4$ AME states need to be two-uniform, while for $n\leq 3$
AME states need to be one-uniform only. This directly implies:

\begin{corollary}
Let $\ket{\psi} \in (\C^d)^{\otimes n}$ be an AME singlet state. 
Then either $n = d = 2$ or $n = d = 3$. For these two cases AME 
singlet states indeed exist as explained in the introduction.
For all other cases of $n$ and $d$ there are no AME singlet 
states.
\label{theorem-qubit}
\end{corollary}

\section{Conclusion}
We have shown that the space of unitarily invariant 
pure quantum states does not contain any two-uniform state. Especially, 
an invariant state cannot be absolutely maximally entangled. 
This generalizes previous results and closes some debates from the 
literature \cite{grassl-njp, sahlmann, fabian-phd}.

There are several ways in which our work can be generalized. 
First, it would be interesting to weaken the condition of 
unitary invariance, by considering a subgroup of the unitary 
group $U(d)$. This may help to understand the set of decoherence 
processes, under which code words of a given quantum code can 
be made robust. Second, the characterization of entangled subspaces 
has attracted recent interest \cite{entsp1, entsp2, entsp3}. In this 
terminology, we characterized the entangled subspace of multiparticle
singlets. It would be very interesting whether similar results can also
be obtained for other constructions of entangled subspaces.

\section{Acknowledgements}
We thank 
Felix Huber,
Balázs Pozsgay,
Hanno Sahlmann,
and
Jonathan Steinberg
for discussions. 
This work was supported by the Deutsche Forschungsgemeinschaft (DFG, German Research 
Foundation, project numbers 447948357 and 440958198), the Sino-German Center 
for Research Promotion (Project M-0294), the ERC (Consolidator Grant 683107/TempoQ) 
and the German Ministry of Education and Research (Project QuKuK, BMBF Grant
No.~16KIS1618K).
FB acknowledges support from the House of Young Talents of the University of Siegen.


\end{document}